\documentclass[10pt]{article}
\usepackage[utf8]{inputenc}
\usepackage[margin=1in]{geometry}
\usepackage{rotating}
\usepackage{booktabs}
\usepackage{array}
\usepackage{amsmath, amsfonts}
\usepackage{parskip}
\usepackage{color}
\usepackage{hyperref}
\usepackage{bbm}
\usepackage{algorithm}
\usepackage[noend]{algpseudocode}
\usepackage{subcaption}
\usepackage{graphbox}

\usepackage{mathtools}
\mathtoolsset{showonlyrefs}

\usepackage[authoryear,round]{natbib}

\usepackage{todonotes}

\usepackage{amsthm}
\newtheorem{theorem}{Theorem}

\usepackage{mathtools}

\def\R{\mathbb{R}}

\title{Incentive-Theoretic Bayesian Inference for Collaborative Science} 
\author{Stephen Bates, Michael I. Jordan, Michael Sklar, Jake A. Soloff} 
\date{\today}

\begin{document}

\maketitle

\begin{abstract}
Contemporary scientific research is a distributed, collaborative endeavor, carried out by teams of researchers, regulatory institutions, funding agencies, commercial partners, and scientific bodies, all interacting with each other and facing different incentives. 
To maintain scientific rigor, statistical methods should acknowledge this state of affairs. To this end, we study hypothesis testing when there is an agent (e.g., a researcher or a pharmaceutical company) with a private prior about an unknown parameter and a principal (e.g., a policymaker or regulator) who wishes to make decisions based on the parameter value. 
The agent chooses whether to run a statistical trial based on their private prior and then the result of the trial is used by the principal to reach a decision.
We show how the principal can conduct statistical inference that leverages the information that is revealed by an agent's strategic behavior---their choice to run a trial or not. 
In particular, we show how the principal can design a policy to elicit partial information about the agent's private prior beliefs and use this to control the posterior probability of the null. 
One implication is a simple guideline for the choice of significance threshold in clinical trials: the type-I error level should be set to be strictly less than the cost of the trial divided by the firm's profit if the trial is successful.
\end{abstract}

\section{Introduction}

Scientific research is increasingly distributed throughout government, academia, and business. Teams of researchers interface with regulatory institutions, funding agencies, commercial partners, scientific bodies, and each other. For example, drug development is conducted by large teams in pharmaceutical companies, with clinical trials carried out in collaboration with academic scientists, whose results are in turn analyzed by public regulators. In such cases, the conclusions from a scientific data analysis materially impact the many stakeholders involved. That is, in addition to its role in quantifying evidence and supporting decision making, statistical analysis serves as a gatekeeping or reward system, with high stakes for the participating parties. Contrast this situation with the classical viewpoint in which statistical protocols are used only as an analytic aid for an impartial researcher. This classical viewpoint underlies the established guidelines for statistical practice. In light of the way statistics is used in present-day research, however, it is essential that we expand the scope of statistical analysis and provide principles that allow economic principles to be wedded with statistical principles in an overall endeavor viewed as a sociotechnical system. 

In any consideration of foundational issues in statistical inference, it is essential to take into account that there are two major, differing perspectives on statistical inference---the frequentist and the Bayesian.  The frequentist paradigm provides guarantees on the correctness of an inference procedure over repeated runs of the procedure.  This paradigm is particularly natural in settings in which a software artifact is built and subsequently used by many individuals on many kinds of data and for many scientific problems.  The Bayesian paradigm focuses on the specific problem at hand, making full use of probability theory to combine past knowledge with current observations, via conditioning and the computation of posterior probabilities.  It provides opportunities for exploiting expert knowledge, requiring effort to elicit such knowledge but potentially rewarding the effort via inferences that can be tailored and sensitive.  Moreover, the Bayesian paradigm accommodates the merging of analyses from multiple investigators who may have different prior distributions and different data~\citep{morris1977combining}, and it also provides for meta-analyses performed by a central aggregator~\citep{sutton2001bayesian}.  These advantages are compelling for the kinds of collaborative science that is our focus, but the need to specify prior distributions remains a stumbling block in many domains.  This is particularly true in complex problem domains in which there are many interacting variables, where it can be difficult to formulate the necessary high-dimensional prior probability distributions, even for domain experts.  This can introduce an undesirable subjectivity and even arbitrariness into scientific decision-making.  Accordingly, the frequentist paradigm, with its focus on confidence intervals and $p$-values, remains the standard in scientific and medical research.  But the limitations of the paradigm are also widely recognized; in particular, it struggles to combine multiple confidence intervals and or $p$-values from multiple sources, and frequentist error control need not translate into good decision-making.

In this work, we show how an incentive-theoretic perspective sheds light on Bayes-frequentist duality and leads to new guidelines for statistical analysis. In particular, we show how a regulator can conduct a Bayesian statistical analysis without assuming a prior distribution. Instead, we view the researcher (who has more information about the subject of study than the regulator) as acting according to an implicit prior distribution and show how the regulator can deduce information about this prior distribution from the researcher's behavior. Loosely speaking, when a researcher makes a large investment in a research undertaking, this credibly signals to the regulator that they have a high degree of belief that the research will be successful, and the regulator can use this information as part of their analysis. 

\subsection*{Overview of our setup}

We consider a setting with two parties: the \textit{principal} (e.g., a regulator such as the FDA) and the \textit{agent} (e.g., a pharmaceutical company). The agent makes an investment to conduct research but must garner approval from the principal. The principal wishes to ensure that only correct conclusions are reached, and they have the ability to approve research results or not. For example, a pharmaceutical company must run a clinical trial to demonstrate that drugs are safe and effective to a regulator who controls whether a new drug is approved.

We view the goal of a scientific study as inferring some aspect of the underlying state of the world from observations. Before conducting the study, we assume that the agent is better informed and has a \emph{prior belief} about the quantity to be inferred. We refer to this quantity as a parameter $\theta$, and we assume that it is a random variable taking values in a sample space $\Theta$ according to a prior distribution $Q$. This distribution is private and not known to the principal. Moreover, the principal cannot simply ask the agent about $Q$, since the agent may engage strategically and report incorrect information for their own benefit. Instead, as a requirement for approving the agent's research, the principal requires that the agent support the conclusions by gathering data (e.g., by running a clinical trial). The agent must decide based on their private information whether or not to invest effort to run such a trial to gain approval. 

In this work, we show how the principal can set up a hypothesis-testing protocol that elicits information about the agent's private prior distribution. This enables the principal to perform Bayesian inference supported by revealed information rather than by assuming a prior distribution.

\section{Bayesian Inference Supported by Revealed Preferences}
\label{sec:using_agent_prior}
\subsection{Setting}
The parameter space $\Theta$ is partitioned into the null set $\Theta_0$ and nonnull set $\Theta_1$. The principal wishes to approve nonnulls but not nulls---this will be formalized shortly. The agent has a prior $Q$ over $\Theta$ that is private and not known to the principal. Noisy evidence about $\theta$ may be available to the principal, however; indeed, the agent may choose to gather data that is visible to both the principal and agent.  We encode the evidence as a random variable $X \in \mathcal{X}$ drawn from the distribution $P_\theta$.

The interaction between the principal and agent goes as follows:
\begin{center}
\fbox{
    \begin{minipage}{5.9in}
    \textbf{Principal-Agent Statistical Trial} \vspace{.01cm} \\
\begin{enumerate}
    \item The principal choose a decision policy $f : \mathcal{X} \to \{\mathsf{approve}, \mathsf{deny}\}$ \vspace{.1cm}
    \item The agent chooses to run a trial at cost $C$, or opts out. \vspace{.1cm}
    \item If the trial is run, evidence is collected according to $X \sim P_\theta$. \vspace{.1cm}
    \item The principal makes a decision $f(X)$. \vspace{.1cm}
    \item If the principal makes the decision $\mathsf{approve}$, the agent receives reward $R$.
\end{enumerate}
\vspace{.1cm}
    \end{minipage}}
\end{center}
\vspace{.1cm}
Putting this together, the agent receives net reward 0 if they opt out, reward $R - C$ if they run a trial that yields decision $\mathsf{approve}$ or reward $-C$ if they run a trial that yields decision $\mathsf{deny}$.
The agent cost $C$ and reward $R$ are known in advance to the principal and the agent.

Without essential loss of generality, we will consider the evidence~$X$ to be a $p$-value for the null hypothesis $\theta \in \Theta_0$; that is, $P_\theta(X \le t) \le t$ for all $\theta \in \Theta_0$. Moreover, we assume that the principal's decision rule is based on thresholding the $p$-value $X$ at a level $\tau$:
\begin{equation}
    f(x) = \begin{cases}
        \mathsf{approve} & x \le \tau \\
        \mathsf{deny} & x > \tau.
    \end{cases}
\end{equation}

We next turn to the agent. Let $\beta_\theta(\tau) = P_\theta(X \le \tau) = P_\theta(\mathsf{approve})$ denote the power function. The agent's expected profit if they run a trial based on their prior $Q$ is then $v_\tau(Q) = \mathbb{E}_{\theta \sim Q}[R \cdot \beta_\theta(\tau)] - C$. The agent additionally has some increasing utility function of their profit $u : \R\to \R$ with $u(0) = 0$, so their expected utility is $\mathbb{E}[u(R\cdot 1\{X\le \tau\} - C)]$. We assume the utility function~$u$ is concave, meaning the agent is risk averse. It follows from Jensen's inequality that the agent has negative expected utility whenever their expected profit is negative. We assume that the agent chooses to opt out when their expected utility is negative.

\subsection{Revealed preferences}

We next show that if the agent acts in their best interest, then the principal can draw conclusions about the posterior probability of the null hypothesis. Our basic assumption is that the agent runs a trial only if $v_\tau(Q) \ge 0$, opting out otherwise. We then have the following:

\begin{theorem}[Incentive-theoretic bound on the posterior odds of null]
\label{thm:posterior_odds}
Suppose the agent runs a trial only if $v_\tau(Q) \ge 0$. Then, when a trial is run the posterior odds of nonnull given approval are bounded from below:
\begin{equation}
\label{eq:post_odds}
    \frac{P\left(\theta \in \Theta_1 \mid \mathsf{approve}\right)}{P\left(\theta \in \Theta_0 \mid \mathsf{approve} \right)} \ge \frac{C / R - \tau}{\tau},
\end{equation}
where the probabilities are according to the agent's prior $Q$ and the randomness in $X$.
\end{theorem}
The proof is elementary and presented in Appendix~\ref{sec:proofs}. This result means that the principal can control the posterior odds of a false discovery by choosing $\tau$.
To see the significance of this, let us rearrange~\eqref{eq:post_odds} to obtain
\begin{equation}
\label{eq:post_null_prob}
P\left(\theta \in \Theta_0 \mid \mathsf{approve}\right)
\le \tau R / C.
\end{equation}
This is a Bayesian form of error control, with the quantity on the left the posterior probability of the null under a decision to $\mathsf{approve}$. This is closely related to the false discovery rate (FDR), the expected fraction of false positives, although the latter is a frequentist statistic, defined with respect to the probability measure induced by repeated sampling without reference to any prior. We comment further on the relationship with FDR in Appendix~\ref{app:fdr_discuss}. Henceforth we use term \emph{Bayes FDR} to refer to the posterior probability in~\eqref{eq:post_null_prob}.

Since the right-hand side of~\eqref{eq:post_null_prob} is decreasing with $\tau$, the principal can guarantee a pre-specified Bayes FDR level $\alpha$ by setting $\tau= \alpha C / R$. The surprising takeaway is that the principal can use the agent's private prior to achieve a desired Bayes FDR level. The principal does not need to have their own prior beliefs about $\theta$ and can instead use the information revealed by the agent (i.e., whether the agent opts in) to learn information about $Q$ and draw conclusions accordingly.

\subsection{Agents with incorrect priors}

We next give another version of this result where the agent's prior is not assumed to be correct. In fact, the agents need not have priors at all. Our result is that no matter what the true values of the unknown parameters are, the fraction of false positives will be small whenever the agents are making money.

\begin{theorem}[Prior-free incentive-theoretic bound]
\label{thm:betting}
Consider $\tau < C/R$. Suppose a set of agents with parameters $\theta^{(1)},\dots,\theta^{(n)}$ opt into the trial above. Then
\begin{equation}
    \underbrace{\sum_{i=1}^n \left(R \cdot P_{\theta^{(i)}}\left(\mathsf{approve}\right) - C\right)}_{\text{expected total profit of agents}} \ge 0
\end{equation}
implies that
\begin{equation*}
    \underbrace{\left(\sum_{i=1}^n \mathbb{I}_{\{\theta^{(i)} \in \Theta_1\}} \cdot P_{\theta^{(i)}}\left(\mathsf{approve}\right)\right)}_{\text{expected number of true positives}} / \underbrace{\left(\sum_{i=1}^n \mathbb{I}_{\{\theta^{(i)} \in \Theta_0\}} \cdot P_{\theta^{(i)}}\left(\mathsf{approve}\right)\right)}_{\text{expected number of false positives}}
    \ge \left(\frac{C / R - \tau}{\tau}\right).
\end{equation*}
Moreover, if the inequality in the first display is strict, the inequality in the second display is also strict.
\end{theorem}

Note that the statement above does not rely on the agents having explicit priors. If an arbitrary collection of agents engages with the principal, then either the fraction of false positives is small or the agents are losing money in aggregate. 

The reader should view Theorem~\ref{thm:posterior_odds} and Theorem~\ref{thm:betting} as two expressions of the same underlying idea. They are linked by the betting interpretation of Bayesian probability; indeed, Theorem~\ref{thm:posterior_odds} means that if the agents are not behaving according to correct beliefs $Q$, they will lose money. Theorem~\ref{thm:betting} is one precise version of this statement.

\subsection{A simple illustration}
\label{sec:simple_example}
We now turn to an explicit example to demonstrate the upper bound. We consider a normal test statistic $Z \sim \mathcal{N}(\theta, 1)$, which is converted into a p-value $X = 1 - \Phi(Z)$, where $\Phi$ is the CDF of the standard normal distribution. We take the parameter space to be $\Theta = \{0, 1\}$ with null set $\Theta_0 = \{0\}$. We suppose there are two types of agents, promising agents and unpromising agents. Promising agents have a prior distribution such that $\theta = 1$ has probability $0.8$, whereas unpromising agents have a prior such that $\theta = 0$ has probability~$1$. Here, the agents' priors are correct. The cost of a trial is 1 and the reward of a successful trial is 100, and we suppose that each agent chooses to run a trial exactly when their expected value is nonnegative. We consider the case where 1\% of the agents are promising agents, which is motivated by clinical trials where nearly all drug candidates are abandoned before conducting a trial, see below for more detail.

We report the fraction of false discoveries across different choices of the p-value threshold ($\tau$) in Figure~\ref{fig:fdr_numerics}. We also report the upper bound from~\eqref{eq:post_null_prob}, which is a valid bound by Theorem~\ref{thm:posterior_odds}.
Notice that the fraction of false discoveries is discontinuous at points where agents change their behavior: promising agents run trials whenever $\tau \ge 0.0005$, and unpromising agents run trials whenever $\tau \ge 0.01$. The upper bound is relatively close to the actual value at these two points, but is somewhat loose in between.

\begin{figure}
    \centering
    \includegraphics[height = 2.5in]{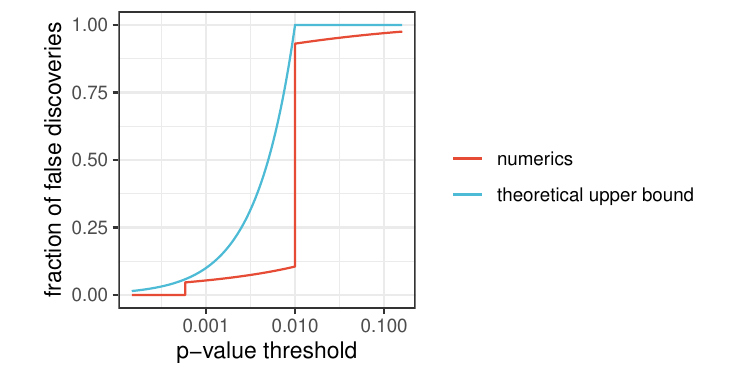}
    \caption{Comparison of the fraction of false discoveries computed numerically to the upper bound from Theorem~\ref{thm:posterior_odds} in the example of Section~\ref{sec:simple_example}.}
    \label{fig:fdr_numerics}
\end{figure}

\section{Implications for Clinical Trials}

We next turn to the choice of significance level when regulating clinical trials, focusing on the US Food and Drug Administration (FDA) as the regulator. A rough sketch of the typical clinical trial pipeline is as follows.
Before a new drug can be marketed, the FDA requires the pharmaceutical company to sponsor a confirmatory clinical trial to establish the drug's safety and efficacy. The costs of the trial are borne by the company, and if the drug is approved the company can then make a profit by selling the drug. Roughly speaking, the drug is approved if the $p$-value from a hypothesis test based on the data from the clinical trial is below a significance threshold---see Section~\ref{subsec:fda-details} for more detail regarding the FDA's current policy.

\subsection{Guidelines from our mathematical model}
What should the significance level be? We suppose that a trial costs $C$ and that the company would be rewarded with a revenue of $R$ if the drug were to be approved. We argue that the type-I error level of the test, $\tau$, should be set to roughly $C/4R$. Our analysis suggests this will result in a fraction of false positives (henceforth referred to as \emph{FDR}) of less than $25\%$ and that this threshold cannot be loosened greatly without resulting in a large FDR. 

This conclusion follows from three facts. First, the work of~\citet{tetenov2016economic} implies that the type-I error level needs to be less than $C/R$, since otherwise companies would be incentivized to run trials for drugs that are not promising or even known a priori to be ineffective. That work notes that there are a large number of drug candidates available that are abandoned before clinical trials are conducted---indeed, there are an estimated 30,000 abandoned candidates for every drug submitted to Phase III clinical trials. If FDA's policy is too loose, many of these unpromising candidates may be submitted to trials. In particular, they conclude that the resulting FDR would be near 1 if the type-I error level is chosen to be greater than $C/R$.

Second, the societal cost of false negatives (effective drugs that are not approved) is typically large compared to false positives (ineffective drugs that are approved), so we should seek a relatively loose FDR level such as 25\%. See the cost-benefit analysis of~\citet{isakov2019fda} for quantitative estimates of the costs of false positives to false negatives. Our conclusion about the right choice of type-I error level is not sensitive to the exact ratio of the societal cost of false positives to false negatives beyond the fact that we should tolerate a moderate fraction of false positives.

Third, our theoretical results suggest that the threshold need not be much smaller than $C/R$. For instance, Theorem~\ref{thm:posterior_odds} and Theorem~\ref{thm:betting} show that a threshold of $C/4R$ suffices to control the FDR at level 25\%. In more detail, we note that pharmaceutical companies incur many research and operating costs other than that of running clinical trials, but their revenues are primarily from selling FDA-approved drugs. Since these companies remain commercially viable, they must be acting in a way such they are not systematically losing money  when deciding whether to sponsor clinical trials. In view of Theorem~\ref{thm:betting}, we conclude that the FDR of approved drugs should be at most $R / C \cdot \tau$. Thus, setting $\tau = C/4R$ would result in an FDR of at most $25\%$.

Of course, the cost $C$ and reward $R$ are different for different drug candidates. Importantly, our analysis suggests that different p-value thresholds should be used for different candidates. We next turn to a detailed analysis of the costs, profits, and existing type-I error levels of clinical trials in the US.

\subsection{The FDA}
\label{subsec:fda-details}

Turning to current FDA policy, regulators have flexibility in analyzing evidence for approval, provided they determine that there is ``substantial evidence'' in cases of approval. In practice, this is usually either two confirmatory trials with positive results (which is generally as a p-value of less than $0.05$), or one multi-center trial with a more significant p-value, such as $0.005$.\footnote{The work of \cite{morant2019characteristics} concludes that from 2012 to 2016 all non-orphan, non-oncology drugs that were approved with a single trial had $p$-values below $0.005$.}
For more information on the criteria employed to select among running one or two trials, see~\cite{food2019demonstrating,haslam2019confirmatory,janiaud2021us}. There is also an accelerated approval process in which a drug may be approved under looser conditions, but we do not analyze that case here.

To analyze the current FDA policy, we consider two simplified statistical protocols that track the current practice. To encapsulate the two-trial case, we consider a protocol such as a drug is approved when two trials show significance at the $0.05$ level. Because a two-sided test is generally used, the resulting probability of a false positive is $0.025$. We call this the \emph{standard} protocol. To handle the one-trial case, we consider a protocol that approves a candidate drug if a single trial is run that yields a p-value below $0.005$. We refer to this as the \emph{modernized} protocol, since it resembles the more streamlined behavior of the FDA Modernization Act of 1997 \nocite{modernization_1997}. The p-value threshold of $0.005$ is based on the findings of~\cite{morant2019characteristics}.

We next turn to the cost of Phase III trials. \citep{moore2018estimated} estimate the median cost to be \$20 million, and \citep{moore2020variation} estimate median total costs to be $\$ 50 $ million. On the other hand, \cite{dimasi2016innovation} estimate average total Phase III costs as $\$255$ million. Lastly, \cite{wouters2020estimated} arrives at an estimate of $\$291$ million. However, these two analyses may overweight large trials. The former uses a private dataset (see the commentary~\cite{love2019criticism}), and the latter is based on costs published in SEC findings. In any case, the costs vary significantly from trial to trial. See \cite{schlander2021much} for a review. To be conservative, we use a value of $C = \$50$ million for total Phase III costs.

Turning to typical values of the reward, $R$, estimates of the total cost of development range from $\$161$ million to $\$4.54$ billion \citep{schlander2021much}, which provides guidance as to the plausible range of drug values. This profitability has a long right tail, with the most commercially successful drugs bringing in around$\$ 100$ billion in revenue~\citep{elmhirst2019biopharma}. We will analyze the case of drugs with profit ranging from $\$ 1$ billion to $\$ 100$ billion.

We report on the FDR bound implied by our theoretical model for the two protocols above in Table~\ref{tab:FDA_realistic}. We find that for typical drugs that would gain \$1 billion to \$10 billion profit if approved, the standard protocol that requires two trials results in a low FDR level. Thus, our analysis suggests the FDA should loosen the significance level in this case. On the other hand, for extremely profitable drugs earning \$100B or more, the protocols are not strict enough, and companies many be incentivized to run clinical trials for unpromising candidates. In this regime, our results do not provide reassurance that the FDR is controlled at a reasonable level.  

Our analysis of the FDA is based on a simplified mathematical model in which the company receives a reward of $R$ even for the approval of drugs that are in fact ineffective. As such, 
an important limitation of our analysis is that it does not account for other regulatory checks and commercial forces (e.g., insurer reluctance to compensate for drugs with less clear benefits, potential lawsuits, and reputational damage) that also work to deter false positives. In general, these checks would mean that our analysis is conservative---the fraction of false positives is perhaps lower than our bound suggests.

Zooming out, we highlight that by bringing together statistical hypothesis testing with private information and incentives, our analysis leads to a richer understanding of the number of false positives one can expect to see in clinical trials. It gives guidance about how one should set the p-value threshold that is tailored to the context, rather than the generic rule-of-thumb that the p-value should be less than $0.05$.

\begin{table}[t]
\centering
\small
\begin{tabular}{|l|r|r|r|r|}
\hline
\multicolumn{1}{|c|}{Protocol} & \multicolumn{1}{c|}{type-I error level ($\tau$)} & \multicolumn{1}{c|}{Revenue if approved ($R$)} & \multicolumn{1}{c|}{Expected profit if null} & \multicolumn{1}{c|}{Bayes FDR bound} \\
\hline
\hline
standard & 0.0625\% & \$1B & -\$49M  & 1.25\% \\
standard & 0.0625\% & \$10B & -\$44M &  12.5\% \\
standard & 0.0625\% & \$100B &  \$13M &  \texttt{n/a} \\
\hline
modernized & 0.5\% & \$1B &  -\$45M & 10\%  \\
modernized & 0.5\% & \$10B &  \$0M & \texttt{n/a} \\
modernized & 0.5\% & \$100B & \$450M & \texttt{n/a} \\
\hline
\end{tabular}

\caption{The behavior of two statistical protocols for varying drug market values, assuming a Phase III cost of $C=\$50$ million. The FDR bound reported in the rightmost column is $R/C \cdot \tau$, which is derived from Theorem~\ref{thm:posterior_odds} and Theorem~\ref{thm:betting}. An entry of \texttt{n/a} in the rightmost column indicates that the FDR cannot be bounded below $100\%$.}
\label{tab:FDA_realistic}
\end{table}

\section{Related Work}

We study the choice of type-I error level in a case where the regulator sets the statistical protocol and the agent decides whether to pay to conduct research, a setting introduced in the insightful work of~\cite{tetenov2016economic}. That work concludes that the optimal type-I error level must be set to be at least as strict as the agent reward divided by cost, otherwise, there may be a large number of false positives. Our work addresses the reverse direction, showing that the type-I error level should not be much stricter than the reward divided by cost. Going beyond the single hypothesis testing setup, \citet{viviano2021when} also consider a regulator setting the protocol for multiple hypothesis tests, and analyze when and how multiple testing adjustments should be carried out. In a different direction,~\citet{bates2022principalagent} show how the regulator can grant partial approval across multiple stages while controlling the number of false positives.

More broadly, the incentives of researchers and their interaction with statistical protocols is a topic commanding increasing attention in econometrics. For example, \cite{chassang2012selective} and \cite{ditillio2017} study randomized trials where an agent has a private action affecting the trial outcome. By contrast, in our work, the statistical trials are not influenced by hidden effort. 
 \citet{spiess2018optimal} studies a principal-agent problem where the agent's action space is the choice of an estimator, showing that the principal restricting the agent to unbiased estimators is worst-case optimal for the principal. \citet{yoder2022designing} study delegation of research to a researcher of unknown efficiency. Similarly,~\citet{mcclellan2022experimentation} derives how the principal should change the approval threshold as evidence is collected sequentially in order to incentivize a researcher to continue experimentation.

One particularly important case of researcher incentives is the $p$-hacking problem, wherein a researcher who is rewarded for positive findings may act in a way that causes false positives.
Classically, \citet{sterling1959publication, tullock1959publication} point out that when only positive findings are published, the total false positive rate of reported results in the scientific literature may be large. 
\citet{leamer1974false} shows how this same phenomenon can occur when a researcher chooses the specification  of a model. 
More recently, there is a vast statistical literature on how an objective analyst can account for multiple comparisons and selective inference~\citep[e.g.,][]{berk2013valid, taylor2015statistical}.
Incorporating incentives more explicitly, \cite{mccloskey2020critical} study how a principal should optimally set the $p$-value threshold when the researcher sequentially collects data and reports a subset of their findings. Similarly, ~\citet{frankel2022findings} study how to select papers for publication based on confidence and effective size. 

There is a growing understanding of persuasion and signaling in the communication of research results from an econometric perspective. An important early study in this direction is~\citet{carpenter2007regulatory}, who consider a game where a researcher can signal their confidence by seeking approval at earlier stages of development.
\citet{andrews2021model}~investigate how research output may be used by multiple downstream decision-makers.
\citet{henry2019research} study a researcher-approver persuasion model where the researcher continuously gathers data. \citet{williams2021preregistration} analyze pre-registration as a form of costly signaling by the analyst, and \citet{banerjee2020theory}~shows how randomized controlled trials emerge from a model with an analyst trying to persuade an adversarial audience. The present work can be situated within this line of thought as a setting in which the agent sends a costly signal by conducting research and the principal can use this signal as part of their statistical analysis.

Lastly, Theorem~\ref{thm:betting} takes inspiration from game-theoretic probability and statistics---see~\citet{shafer2005probability} and~\citet{ramdas2022game} for overviews. See the Appendix therein for a discussion of the philosophical underpinnings of game-theoretic statistics, especially in relation to frequentist and Bayesian paradigms. Game-theoretic statistics adopts the language of betting to quantify uncertainty. Evidence against a null hypothesis is measured by the outcome of a bet---one which is {\sl fair} in the sense that it would not be profitable under the null hypothesis. \citet{shafer2021testing} advocates for betting terminology as an effective framework for communicating results to a broad audience. In our work, of course, the agent has an actual financial stake in the outcome of the experiment. The betting analogy delivers another key benefit to game-theoretic statistics: bets can be made sequentially, so evidence can be accumulated gradually within an experiment or aggregated across studies. In our current work, however, the agent computes a $p$-value, effectively placing an all-or-nothing bet against the null hypothesis. \citet{bates2022principalagent} show how the principal can employ more general betting scores (also known as $e$-values) to align the incentives of the agent.

\section{Discussion}

We have shown how the structure of economic incentives can support statistical inference.  Our work provides a new connection between frequentist error control (e.g., controlling the false discovery rate) and Bayesian statistics. A frequent criticism of Bayesian statistical methods is their reliance on a prior distribution. Our work uses a prior distribution, but in a way that is objectively verifiable---in Theorem~\ref{thm:betting} we see that if the agent priors are invalid, then we would observe that the agents are losing money.
In fact, a closer inspection of the proof reveals that if FDR is not controlled, the agents would have a negative reward linear in the number of trials, i.e., a large amount of money. In settings such as clinical trials, we know that this is not the case, which lends credence to the use of the agents' implicit prior distributions as a basis for inference. More generally, accounting for and leveraging the broader economic context of statistical analysis is increasingly important as the process of research becomes more complex with many interacting, strategic stakeholders, and our work is one step in this direction.

\section*{Acknowledgements}
We thank Jon McAuliffe and Aaditya Ramdas for helpful discussions.

\bibliographystyle{plainnat}
\bibliography{incentives}

\newpage
\appendix
\section{Proofs and Additional Formal Results}\label{sec:proofs}
\begin{proof}[Proof of Theorem~\ref{thm:posterior_odds}]
First, note that since the agent opted in to the trial, the agent's prior $Q$ satisfies
\begin{equation}
\label{eq:participation_constr}
\mathbb{E}_{\theta \sim Q}[\beta_\theta(\tau)] \ge C / R.
\end{equation}
Next, we have
\begin{align*}
        \frac{P\left(\theta \in \Theta_1 \mid \mathsf{approve}\right)}{P\left(\theta \in \Theta_0 \mid \mathsf{approve} \right)} 
        &= \frac{P(\theta \in \Theta_1 \text{ and } X \le \tau) / P(X \le \tau)}{P(\theta \in \Theta_0 \text{ and } X \le \tau) / P(X \le \tau)} \\ 
    &= \frac{\mathbb{E}[\beta_\theta(\tau) \mathbb{I}_{\theta \in \Theta_1}]}{\mathbb{E}[\beta_\theta(\tau) \mathbb{I}_{\theta \in \Theta_0}]}.
\end{align*}
Denote the numerator by $a$ and the denominator by $b$. By~\eqref{eq:participation_constr}, we have that $a + b \ge C/R$. Thus, from the above display we have $a / b \ge (C/R - b) / b$, which is decreasing in $b$. Finally, note that $b \le \tau$ since $X$ is a valid $p$-value, yielding the desired result.
\end{proof}

\begin{proof}[Proof of Theorem~\ref{thm:betting}]
    For $i=1,\dots,n$, consider the following:
\begin{align*}
    v^p_i &= \mathbb{I}_{\{\theta^{(i)} \in \Theta_1\}} \cdot \beta_{\theta^{(i)}}(\tau) - \left(\frac{C / R - \tau}{\tau}\right) \cdot \mathbb{I}_{\{\theta^{(i)} \in \Theta_0\}} \cdot \beta_{\theta^{(i)}}(\tau) \\
    v^a_i &= \beta_{\theta^{(i)}}(\tau) - C/R.
\end{align*}
The reader should think of $v^p_i$ as the expected value of a game for the principal when interacting with an agent with parameter $\theta^{(i)}$ and $v^a_i$ as the expected value of the same game for the agent.

Next, we claim that $v^p_i \ge v^a_i$ for all $i$. This is immediate for $i$ such that $\theta^{(i)} \in \Theta_1$. To verify that $v^p_i \ge v^a_i$ when $\theta^{(i)} \in \Theta_0$, note that $\beta_{\theta^{(i)}}(\tau) \le \tau$, since $X$ is a $p$-value under the null. Thus, $v^p_i \ge - (C/R - \tau)$ and $v^a_i \le \tau - C/R$. We conclude $v^p_i \ge v^a_i$ for all $i$.

The first inequality in the desired theorem statement is equivalent to $\sum_i v_i^a \ge 0$, and the second inequality is equivalent to $\sum_i v_i^p \ge 0$. Thus, the desired result follows from the fact that $v^p_i \ge v^a_i$.
\end{proof}

\section{Discussion of the False Discovery Rate and its Variants}
\label{app:fdr_discuss}

In this section, we review various definitions of the false discovery rate (FDR) and show how they relate to the error control guarantee in Theorems~\ref{thm:posterior_odds} and~\ref{thm:betting}. FDR was first defined in the context of multiple hypothesis testing, where the aim is to control the expected proportion of false discoveries among all tests declared statistically significant. \citet{seeger1968note} attributes this definition (with different terminology) to a series of unpublished papers by G. Eklund from 1961-1963; the FDR was independently rediscovered and popularized by the seminal work of \citet{benjamini1995controlling}---see \citet{benjamini2000adaptive} for a review of the history of this concept in multiple hypothesis testing. 

The FDR is closely related to older concepts in Bayesian testing and classification theory---some fundamental connections were established by \citet{efron2001empirical,genovese2002operating,storey2003positive}. The Bayesian formulation makes clear the relevance of FDR when testing even a single hypothesis, so we will focus on a formal exposition of the Bayesian approach here. Our terminology follows \cite{efron2012large}.

Building upon the notation from Section~\ref{sec:using_agent_prior}, define the null prior probability $\pi_0 := Q(\Theta_0)$. For $j\in\{0,1\}$, let $F_j$ and $f_j$ denote the conditional distribution function and density function, respectively, of $X$ given $\theta\in \Theta_j$. By the law of total probability, the marginal distribution of $X$ is 
\[
F_X(\tau) := P(X\le \tau) = \pi_0 F_0(\tau) + (1-\pi_0)F_1(\tau),
\]
and similarly the marginal density $f_X$ is a mixture of the null and alternative densities $f_0$ and $f_1$, respectively. For any subset $\mathcal{Z}\subseteq \mathcal{X}$ of the sample space, the \emph{Bayes FDR} over $\mathcal{Z}$ is defined as the posterior probability of the null, given $X\in \mathcal{Z}$:
\[
\text{Fdr}(\mathcal{Z}) := P(\theta\in \Theta_0 \mid X\in \mathcal{Z}).
\]
For instance, if $\mathcal{Z}_\mathsf{approve} = [0,\tau]$ denotes the set of values of $X$ for which the principal decides $\mathsf{approve}$, then the corresponding Bayes FDR
becomes 
\[
\text{Fdr}(\mathcal{Z}_\mathsf{approve}) = P(\theta\in \Theta_0 \mid \mathsf{approve}) = \frac{\pi_0F_0(\tau)}{F_X(\tau)},
\]
which is the posterior probability~\eqref{eq:post_null_prob} from Theorem~\ref{thm:posterior_odds}. Similarly, the quantity in the conclusion of Theorem~\ref{thm:betting} is a rearrangement of the \emph{marginal false discovery rate}---the ratio of the expected number of false positives to the expected number of rejections.

Of course, the principal and agent do not only observe whether $X\in \mathcal{Z}_\mathsf{approve}$---they observe the exact value of~$X$. An even more relevant quantity, known as the \emph{local false discovery rate}, is defined as the posterior probability of the null given $X$ equals some constant $\tau$:
\[
\text{lfdr}(\tau)
:= P(\theta\in\Theta_0\mid X=\tau)
= \frac{\pi_0f_0(\tau)}{f_X(\tau)}.
\]
As the Bayesian false discovery rate $\text{Fdr}(\mathcal{Z})$ is defined for any region $\mathcal{Z}\subseteq\mathcal{X}$, the local false discovery rate $\text{lfdr}(\tau)$ is defined for any $\tau\in \mathcal{X}$. Indeed, the two are related by $\text{lfdr}(\tau) = \text{Fdr}(\{\tau\})$. An interesting open question is whether some analogue of Theorem~\ref{thm:posterior_odds} holds for the local false discovery rate under more assumptions.

\end{document}